\documentclass[copyright]{eptcs}
\usepackage{breakurl}             

\usepackage{amsmath, amssymb, latexsym, amsthm}

\DeclareMathOperator{\var}{var}

\theoremstyle{plain}
\newtheorem{problem}{Problem}
\newtheorem{theorem}{Theorem}
\newtheorem{lemma}{Lemma}
\newtheorem{corollary}{Corollary}
\newtheorem{proposition}{Proposition}

\newtheorem{conjecture}{Conjecture}
\newtheorem{definition}{Definition}

\title{Unambiguous 1-Uniform Morphisms}
\author{Hossein Nevisi\thanks{Corresponding author.} \qquad \qquad Daniel Reidenbach
\institute{Department of Computer Science\\
Loughborough University\\
Loughborough, LE11 3TU, United Kingdom}
\email{H.Nevisi@lboro.ac.uk \qquad D.Reidenbach@lboro.ac.uk}
}
\def\titlerunning{Unambiguous 1-Uniform Morphisms}
\def\authorrunning{H.~Nevisi \& D.~Reidenbach}
\begin{document}
\maketitle

\begin{abstract}
A morphism $\sigma$ is unambiguous with respect to a word $\alpha$ if there is no other morphism $\tau$ that maps $\alpha$ to the same image as $\sigma$. In the present paper we study the question of whether, for any given word, there exists an unambiguous 1-uniform morphism, i.\,e., a morphism that maps every letter in the word to an image of length $1$.
\end{abstract}

\section{Introduction}\label{sec:intro}

If, for a morphism $\sigma: \Delta^* \to \Sigma^*$ (where $\Delta$ and $\Sigma$ are arbitrary alphabets) and a word $\alpha \in \Delta^*$, there exists another morphism $\tau$ mapping $\alpha$ to $\sigma(\alpha)$, then $\sigma$ is called \emph{ambiguous} with respect to $\alpha$; if such a $\tau$ does not exist, then $\sigma$ is \emph{unambiguous}. For example, the morphism $\sigma_0: \{A,B,C\}^* \to \{a,b\}^*$ -- given by $\sigma_0(A) := a$, $\sigma_0(B) := a$, $\sigma_0(C) := b$ -- is ambiguous with respect to the word $\alpha_0 := ABCACB$, since the morphism $\tau_0$ -- defined by $\tau_0(A) := \varepsilon$ (i.\,e., $\tau_0$ maps $A$ to the empty word), $\tau_0(B) := a$, $\tau_0(C) := ab$ -- satisfies $\tau_0(\alpha_0) = \sigma_0(\alpha_0)$ and, for a symbol $X$ occuring in $\alpha$, $\tau_0(X) \neq \sigma_0(X)$:
\begin{eqnarray*}
 \sigma_0(\alpha_0) = & \overbrace{\phantom{bb}a\phantom{bb}}^{\sigma_0(A)} \overbrace{\phantom{bb}a\phantom{bb}}^{\sigma_0(B)} \overbrace{\phantom{bb}b\phantom{bb}}^{\sigma_0(C)} \overbrace{\phantom{bb}a\phantom{bb}}^{\sigma_0(A)} \overbrace{\phantom{bb}b\phantom{bb}}^{\sigma_0(C)} \overbrace{\phantom{bb}a\phantom{bb}}^{\sigma_0(B)} & = \tau_0(\alpha_0) \, .\\[-1.5em]
\hphantom{ \sigma_0(\alpha_0) =} & \underbrace{\hphantom{bbabb}}_{\tau_0(B)} \underbrace{\hphantom{bbbbabbbbb}}_{\tau_0(C)} \underbrace{\hphantom{bbbbabbbbb}}_{\tau_0(C)} \underbrace{\hphantom{bbabb}}_{\tau_0(B)} & \hphantom{\tau_0(\alpha_0) \, .}
\end{eqnarray*}
It can be verified with moderate effort that, e.\,g., the morphism $\sigma_1: \{A,B,C\}^* \to \{a,b\}^*$ -- given by $\sigma_1(A) := a$, $\sigma_1(B) := ab$, $\sigma_1(C) := b$ -- is unambiguous with respect to $\alpha_0$.\par
The potential ambiguity of morphisms is relevant to various concepts in the combinatorial theory of morphisms, such as pattern languages (see, e.\,g., Mateescu and Salomaa~\cite{mat:pat}), equality sets (see, e\,g., Harju and Karhum\"aki~\cite{har:mor}) and word equations (see, e.\,g., Choffrut~\cite{cho:equ}). This relation is best understood for inductive inference of pattern languages, where it has been shown that a preimage can be computed from some of its morphic images if and only if these images have been generated by morphisms with a restricted ambiguity (see, e.\,g., Reidenbach~\cite{rei:dis}). Hence, intuitively speaking, unambiguous morphisms have a desirable, namely structure-preserving, property in such a context, and therefore previous literature on the ambiguity of morphisms mainly studies the question of the \emph{existence} of unambiguous morphisms for arbitary words. In the initial paper, Freydenberger, Reidenbach and Schneider~\cite{fre:una2} show that there exists an unambiguous nonerasing morphism with respect to a word $\alpha$ if and only if $\alpha$ is not a \emph{fixed point} of a nontrivial morphism, i.\,e., there is no morphism $\phi$ satisfying $\phi(\alpha) = \alpha$ and, for a symbol $x$ in $\alpha$, $\phi(x) \neq x$. Freydenberger and Reidenbach~\cite{fre:the2} study those sets of words with respect to which so-called segmented morphisms are unambiguous, and these results lead to a refinement of the techniques used in \cite{fre:una2}. Schneider~\cite{sch:una2} and Reidenbach and Schneider~\cite{rei:res2} investigate the existence of unambiguous erasing morphisms -- i.\,e., morphisms that may map symbols to the empty word. Finally, Freydenberger, Nevisi and Reidenbach~\cite{fre:wea} study a definition of unambiguity that is completely restricted to nonerasing morphisms\footnote{Note that \cite{fre:una2,fre:the2} also deal with unambiguous \emph{nonerasing} morphisms, but they use a stronger notion of unambiguity that is based on arbitrary monoid morphisms. Hence, they call a morphism $\sigma$ unambiguous only if there is no other -- erasing or nonerasing -- morphism $\tau$ satisfying $\tau(\alpha) = \sigma(\alpha)$. In contrast to this, and in contrast to the present paper, \cite{fre:wea} disregards erasing morphisms $\tau$. Consequently, in the definition of unambiguity studied by \cite{fre:wea}, our initial example $\sigma_0$ is considered (``weakly'') \emph{unambiguous} with respect to $\alpha_0$, since all morphisms $\tau$ with $\tau(\alpha_0) = \sigma_0(\alpha_0)$ are erasing morphisms.}, and they provide a characterisation of those words with respect to which there exist unambiguous morphisms $\sigma: \Delta^+ \to \Sigma^+$ in such a context (this characterisation does not hold for binary target alphabets $\Sigma$, though).\par
In the present paper, we study the existence of unambiguous \emph{1-uniform} morphisms for arbitrary words, i.\,e., just as our initial example $\sigma_0$, these morphisms map every symbol in the preimage to an image of length $1$. In order to obtain unrestricted results, we wish to consider words over an unbounded alphabet $\Delta$ as morphic preimages. Therefore, we assume $\Delta := \mathbb{N}$; in accordance with the existing literature in the field, we call any word $\alpha \in \mathbb{N}^*$ a \emph{pattern}, and we call any symbol $x \in \mathbb{N}$ occurring in $\alpha$ a \emph{variable}. Thus, more formally, we wish to investigate the following problem:
\begin{problem}\label{prob:main}
Let $\alpha \in \mathbb{N}^*$ be a pattern, and let $\Sigma$ be an alphabet. Does there exists a 1-uniform morphism $\sigma: \mathbb{N}^* \to \Sigma^*$ that is unambiguous with respect to $\alpha$, i.\,e., there is no morphism $\tau: \mathbb{N}^* \to \Sigma^*$ satisfying $\tau(\alpha) = \sigma(\alpha)$ and, for a variable $x$ occurring in $\alpha$, $\tau(x) \neq \sigma(x)$?
\end{problem}
There are two main reasons why we study this question: Firstly, any insight into the existence of unambiguous 1-uniform morphisms improves the construction by Freydenberger et al.~\cite{fre:una2}, which provides comprehensive results on the existence of unambiguous nonerasing morphisms, but is based on morphisms that are often much more involved than required. This can be illustrated using our above example pattern $\alpha_0$ (now interpreted as $\alpha_0 := 1 \cdot 2 \cdot 3 \cdot 1 \cdot 3 \cdot 2$ in order to fit with the definition of patterns as words over $\mathbb{N}$). Here, the unambiguous morphism $\sigma_1$ -- which is not 1-uniform, but still of very limited complexity -- produces a morphic image of length $8$, whereas the unambiguous morphism for $\alpha_0$ defined in \cite{fre:una2} leads to a morphic image of length $162$. This substantial complexity of known unambiguous morphisms has a severe effect on the runtime of inductive inference procedures for pattern languages, which, as mentioned above, are necessarily based on such morphisms. Thus, any insight into the existence of uncomplex unambiguous morphisms is not only of intrinsic interest, but is also important from a more applied point of view. Secondly, as shown by $\sigma_0(\alpha_0)$, the images under 1-uniform morphisms have a structure that is very close to that of their preimages. This is because, whenever the pattern contains more different variables than there are letters in the target alphabet, a 1-uniform morphism reduces the complexity of the preimage by mapping certain variables to the same image. Thus, such a \emph{morphic simplification} and its potential ambiguity are a very basic phenomenon in the combinatorial theory of morphisms. Our studies shall suggest that Problem~\ref{prob:main} is nevertheless a challenging question, and we shall demonstrate that it is related to a number of other concepts and problems in combinatorics on words.\par
Note that, due to space constraints, this extended abstract contains just a few proofs, focussing on those that are reasonably short and suitable to illustrate our basic proof techniques.

\section{Definitions and Preliminary Results}\label{sec:def}

For the definitions of \emph{patterns}, \emph{variables}, \emph{1-uniform morphisms}, \emph{(un)ambiguous morphisms}, \emph{fixed points} of nontrivial morphisms, and the symbol $\varepsilon$, Section~\ref{sec:intro} can be consulted.\par
Let $A$ be an \emph{alphabet}, i.\,e., an enumerable set of symbols. A \emph{word (over $A$)} is a a finite sequence of symbols taken from $A$. The set $A^*$ is the set of all words over $A$, and $A^+ := A^* \setminus \{ \varepsilon \}$. For the \emph{concatenation} of two words $w_1, w_2$, we write $w_1 \cdot w_2$ or simply $w_1 w_2$. The notion $|x|$ stands for the size of a set $x$ or the length of a word $x$. For any word $w \in A^*$, the notation $|w|_x$ stands for the number of occurrences of the letter $x$ in $w$. The symbol $[\ldots]$ is used to omit some canonically defined parts of a given word, e.\,g., $\alpha = 1 \cdot 2 \cdot [\ldots] \cdot 5$ stands for $\alpha = 1 \cdot 2 \cdot 3 \cdot 4 \cdot 5$. We call a word $v \in A^*$ a \emph{factor} of a word $w \in A^*$ if, for some $u_1, u_2 \in A^*$, $w = u_1 v u_2$; moreover, if $v$ is a factor of $w$ then we say that $w$ contains $v$ and denote this by $v \sqsubseteq w$ or $w = \cdots v \cdots$. If $v \neq w$, then we say that $v$ is a \emph{proper} factor of $w$ and denote this by $v \sqsubset w$. If $u_1 = \varepsilon$, then $v$ is a \emph{prefix} of $w$, and if $u_2 = \varepsilon$, then $v$ is a \emph{suffix} of $w$. For every letter $x$ in $w$, $L_x := \{ y \in A \mid w = \cdots y \cdot x \cdots \} \cup L'_x$ and $R_x := \{ y \in A \mid w = \cdots x \cdot y \cdots \} \cup R'_x$, where $L'_x = \{ \varepsilon \}$ if $w = x \cdots$ and  $L'_x = \emptyset$ if $w \neq x \cdots$, and $R'_x = \{ \varepsilon \}$ if $w = \cdots x$ and  $R'_x = \emptyset$ if $w \neq \cdots x$. We refer to the sets $L_x$ and $R_x$ as \emph{neighbourhood sets}.\par
For alphabets $A,B$, a mapping $h: A^* \to B^*$ is a \emph{morphism} if $h$ is compatible with the concatenation, i.\,e., for all $v,w \in A^*$, $h(v) \cdot h(w) = h(vw)$. We call $B$ the \emph{target alphabet} of $h$. The morphism $h$ is said to be \emph{nonerasing} if, for every $x \in A$, $h(x) \neq \varepsilon$. A morphism is called a \emph{renaming} if it is injective and 1-uniform. We additionally call any word $v$ a renaming of a word $w$ if there is a morphism $h$ that is a renaming and satisfies $h(w) = v$. A word $w \in A^*$ is said to be \emph{in canonical form} if it is lexicographically minimal (with regard to any fixed order on $A$) among all its renamings in $A^*$. \par
With regard to an arbitrary pattern $\alpha \in \mathbb{N}^*$, $\var(\alpha)$ denotes the set of all variables occurring in $\alpha$. If we say that a pattern is in canonical form, then this shall always refer to the usual order on $\mathbb{N}$, i.\,e., $1 < 2 < 3 < \ldots$\,.\par
The question of whether a pattern $\alpha$ is a fixed point of a nontrivial morphism (which can be decided in polynomial time, see Holub~\cite{hol:pol}) is equivalent to a number of other concepts in combinatorics on words. More precisely, $\alpha$ is a fixed point of a nontrivial morphism iff $\alpha$ is \emph{prolix} iff $\alpha$ is \emph{morphically imprimitive} iff there exist a certain characteristic factorisation of $\alpha$; these equivalences are explained by Reidenbach and Schneider~\cite{rei:mor2} in more detail. Results on unambiguous morphisms have been stated using any of these concepts. In the present paper, our presentation shall focus on the notion of fixed points. Therefore, we can now paraphrase a simple yet fundamental insight by Freydenberger et al.~\cite{fre:una2} -- which implies that an answer to Problem~\ref{prob:main} is trivial for those patterns that are fixed points of nontrivial morphisms -- as follows:
\begin{theorem}[Freydenberger et al.~\cite{fre:una2}]\label{thm:prolix_ambig}
Let $\alpha \in \mathbb{N}^*$ be a fixed point of a nontrivial morphisms, and let $\Sigma$ be any alphabet. Then every nonerasing morphism $\sigma: \mathbb{N}^* \to \Sigma^*$ is ambiguous with respect to $\alpha$.
\end{theorem}\noindent
Hence, we can safely restrict our subsequent considerations to those patterns that are not fixed points.

\section{Fixed Target Alphabets}\label{sec:fixed}

In the the present section, we describe a number of conditions on the existence of unambiguous 1-uniform morphisms $\sigma: \mathbb{N}^* \to \Sigma^*$ with a \emph{fixed} target alphabet $\Sigma$, i.\,e., the size of $\Sigma$ does not depend on the number of variables occurring in $\alpha$. While the main result by Freydenberger et al.~\cite{fre:una2} demonstrates that the set of patterns with an unambiguous \emph{nonerasing} morphisms is independent of the size of $\Sigma$ (provided that $|\Sigma| \geq 2$), our initial example $\alpha_0$ and all patterns $\alpha_m := 1 \cdot 1 \cdot 2 \cdot 2 \cdot [\ldots] \cdot m \cdot m$ with $m \geq 4$ do not have an unambiguous \emph{1-uniform} morphism $\sigma: \mathbb{N}^* \to \Sigma^*$ for binary alphabets $\Sigma$. In contrast to this, such morphisms can be given for ternary (and, thus, larger) alphabets:
\begin{theorem}\label{thm:Thue}
Let $m\in\mathbb{N}$, $m\geq 4$, let $\Sigma$ be an alphabet, and let $\alpha_m := 1\cdot 1 \cdot 2\cdot 2\cdot [\ldots]\cdot m\cdot m$. There exists a 1-uniform morphism $\sigma: \mathbb{N}^* \to \Sigma^*$ that is unambiguous with respect to $\alpha_m$ if and only if $|\Sigma| \geq 3$.
\end{theorem}

\begin{proof}
Since squares cannot be avoided over unary and binary alphabets, it can be shown with very limited effort that there is no unambiguous 1-uniform morphism $\sigma: \mathbb{N}^* \to \Sigma^*$ with respect to any $\alpha_m$ if $\Sigma$ does not contain at least three letters.\par
According to Thue~\cite{thu:ueb1}, there exists an infinite square-free word over a ternary alphabet. Let this word be $w$. Thus,
\begin{displaymath}
w=abcacbabcbacabcacbaca\cdots\:.
\end{displaymath}
We define the word $w'$ by repeating every letter of $w$ twice. Consequently,
\begin{displaymath}
w'=aabbccaaccbbaabbccbbaaccaabbccaaccbbaaccaa\cdots\:.
\end{displaymath}
We now define a 1-uniform morphism $\sigma: \mathbb{N}^*\to \{a,b,c\}^*$ such that $\sigma(\alpha_m)$ is a prefix of $w'$. Since $w$ is square-free, the only square factors of $w'$ are $aa$, $bb$ and $cc$. Hence, it can be easily verified that $\sigma$ is unambiguous with respect to $\alpha_m$.
\end{proof}

Thus -- and just as for the equivalent problem on unambiguous \emph{erasing} morphisms (see Schneider~\cite{sch:una2}) -- any characteristic condition on the existence of unambiguous 1-uniform morphisms needs to incorporate the size of $\Sigma$, which suggests that such criteria might be involved. Therefore, our results in this section are restricted to \emph{sufficient} conditions on the existence of unambiguous 1-uniform morphisms.\par
Our first criterion is based on (un)avoidable patterns and is, thus, related to the above-mentioned property of the patterns $\alpha_m$:
\begin{theorem}\label{thm:noncross}
Let $n\in\mathbb{N}$, $\beta := r_1\cdot r_2\cdot[\ldots]\cdot r_{\lceil n/2\rceil}$ and $\alpha := 1^{r_1}\cdot 2^{r_1}\cdot 3^{r_2}\cdot 4^{r_2}\cdot[\ldots]\cdot n^{(r_{\lceil n/2\rceil})}$ with $r_i \geq 2$ for every $i$, $1\leq i\leq \lceil n/2\rceil$. If $\beta$ is square-free, then there exists a 1-uniform morphism $\sigma : \mathbb{N}^*\to \{a,b\}^*$ that is unambiguous with respect to $\alpha$.
\end{theorem}

Our second criterion again holds for binary (and, thus, all larger) alphabets $\Sigma$. It features a rather restricted class of patterns, which, however, are minimal with regard to their length.
\begin{theorem}\label{thm:shortest_succinct_binary}
Let $n\in\mathbb{N}$, $n \geq 2$. If $n$ is even, let
\begin{displaymath}
\alpha := 1\cdot2\cdot[\ldots]\cdot n\cdot(n/2+1)\cdot 1\cdot (n/2+2)\cdot 2\cdot[\ldots]\cdot n\cdot n/2,
\end{displaymath}
and if $n$ is odd, let
\begin{displaymath}
\alpha := 1\cdot1\cdot2\cdot3\cdot[\ldots]\cdot n\cdot(\lceil n/2\rceil+1)\cdot 2\cdot (\lceil n/2\rceil+2)\cdot 3\cdot[\ldots]\cdot n\cdot \lceil n/2\rceil.
\end{displaymath}
Then $\alpha$ is a shortest pattern with $|\var(\alpha)| = n$ that is not a fixed point of a nontrivial morphism, and there exists a 1-uniform morphism $\sigma: \mathbb{N}^* \to \{a,b\}^*$ that is unambiguous with respect to $\alpha$.
\end{theorem}

The following examples illustrates Theorem~\ref{thm:shortest_succinct_binary} and its proof: For $n := 6$, $\alpha := 1 \cdot 2 \cdot 3 \cdot 4 \cdot 5 \cdot 6 \cdot 4 \cdot 1 \cdot 5 \cdot 2 \cdot 6 \cdot 3$, and the 1-uniform morphism $\sigma:\mathbb{N}^* \to \{a,b\}^*$ with $\sigma(1) := \sigma(2) := \sigma(3) := a$ and $\sigma(4) := \sigma(5) := \sigma(6) := b$ is unambiguous with respect to $\alpha$. For $n := 5$, $\alpha := 1 \cdot 1 \cdot 2 \cdot 3 \cdot4 \cdot 5 \cdot4 \cdot 2 \cdot 5 \cdot 3$, and the respective unambiguous morphism is given by $\sigma(1) := \sigma(2) := \sigma(3) := a$ and $\sigma(4) := \sigma(5) :=b$.\par
From Theorem~\ref{thm:shortest_succinct_binary} we can conclude that patterns $\alpha$ with unambiguous 1-uniform morphisms using a binary target alphabet exist for every cardinality of $\var(\alpha)$ and that corresponding examples can be given where every variable occurs just twice.

\section{Variable Target Alphabets}\label{sec:variable}

In order to continue our examination of Problem~\ref{prob:main}, we now relax one of the requirements of Section~\ref{sec:fixed}: We no longer investigate criteria on the existence of unambiguous 1-uniform morphisms for a fixed target alphabet $\Sigma$, but we permit $\Sigma$ to depend on the number of variables in the pattern $\alpha$ in question. Regarding this question, we conjecture the following statement to be true:
\begin{conjecture}\label{conj:variable}
Let $\alpha$ be a pattern with $|\var(\alpha)| \geq 4$. There exists an alphabet $\Sigma$ satisfying $|\Sigma| < |\var(\alpha)|$ and a 1-uniform morphism $\sigma: \mathbb{N}^* \to \Sigma^*$ that is unambiguous with respect to $\alpha$ if and only if $\alpha$ is not a fixed point of a nontrivial morphism.
\end{conjecture}\noindent
This conjecture would be trivially true if we allowed $\Sigma$ to satisfy $|\Sigma| \geq |\var(\alpha)|$. That explains why we exclusively study the case where the number of letters in the target alphabet is smaller than the number of variables in the pattern. From Theorem~\ref{thm:Thue}, it directly follows that an analogous conjecture would not be true if we considered fixed binary target alphabets (as is done in Section~\ref{sec:fixed}), since none of the patterns $\alpha_m$ is a fixed point of a nontrivial morphism -- this can be easily verified using tools discussed by Reidenbach and Schneider~\cite{rei:mor2} and Holub~\cite{hol:pol}. Hence, characteristic criteria must necessarily look different in such a context. It can also be effortlessly understood that Conjecture~\ref{conj:variable} would be incorrect if we dropped the condition that $\alpha$ needs to contain at least $4$ distinct variables, since not only $\sigma_0$, but all 1-uniform morphisms $\sigma: \mathbb{N}^* \to \Sigma^*$ with $|\Sigma| \leq 2$ are ambiguous with respect to our example pattern $\alpha_0 = 1 \cdot 2 \cdot 3 \cdot 1 \cdot 3 \cdot 2$ discussed in Section~\ref{sec:intro}.\par
Technically, many of our subsequent technical considerations are based on the following generic morphisms:
\begin{definition}\label{def:sigij}
Let $\Sigma$ be an infinite alphabet, and let $\sigma: \mathbb{N}^* \to \Sigma^*$ be a renaming. For any $i,j \in \mathbb{N}$ with $i \neq j$ and for every $x \in \mathbb{N}$, let the morphism $\sigma_{i,j}$ be given by
\begin{displaymath}
\sigma_{i,j}(x) :=
\begin{cases}
\sigma(i), & \text{if }\, x = j \, ,\\
\sigma(x), & \text{if }\, x \neq j \, .
\end{cases}
\end{displaymath}
\end{definition}\noindent
Thus, $\sigma_{i,j}$ maps exactly two variables to the same image, and therefore, for any pattern $\alpha$ with at least two different variables, $\sigma_{i,j}(\alpha)$ is a word over $|\var(\alpha)| - 1$ distinct letters. Using this definition, we can now state a more specific version of Conjecture~\ref{conj:variable}:
\begin{conjecture}\label{conj:sigij}
Let $\alpha$ be a pattern with $|\var(\alpha)| \geq 4$. There exist $i,j \in \var(\alpha)$, $i \neq j$, such that $\sigma_{i,j}$ is unambiguous with respect to $\alpha$ if and only if $\alpha$ is not a fixed point of a nontrivial morphism.
\end{conjecture}
As a side note, we consider it worth mentioning that Conjecture~\ref{conj:sigij} shows connections to another conjecture from the literature. In order to state the latter, we define, for any $i \in \mathbb{N}$, the morphism $\delta_i: \mathbb{N}^* \to \mathbb{N}^*$ by $\delta_i(i):= \varepsilon$ and, for every $j \in \mathbb{N} \setminus \{i\}$, $\delta_i(j) := j$.
\begin{conjecture}[Billaud~\cite{bil:apr}, Lev\'e and Richomme~\cite{lev:ona}]\label{conj:Billaud}
Let $\alpha$ be a pattern with $|\var(\alpha)| \geq 3$. If, for every $i \in \var(\alpha)$, $\delta_i(\alpha)$ is a fixed point of a nontrivial morphism, then $\alpha$ is a fixed point of a nontrivial morphism.
\end{conjecture}\noindent
In general, the correctness of Conjecture~\ref{conj:Billaud} has not been established yet. The problem is intensively studied by Lev\'e and Richomme~\cite{lev:ona}, where it is shown to be correct for certain subclasses of $\mathbb{N}^*$.\par
Due to Theorem~\ref{thm:prolix_ambig}, the \emph{only if} directions of Conjectures~\ref{conj:variable} and~\ref{conj:sigij} hold true immediately. In the remainder of this section, we shall therefore exclusively study those patterns that are not fixed points. Our corresponding results yield large classes of such patterns that have an unambiguous 1-uniform morphism, but we have to leave the overall correctness of our conjectures open.\par
Conjecture~\ref{conj:sigij} suggests that the examination of the existence of unambiguous 1-uniform morphisms for a pattern $\alpha$ may be reduced to finding suitable variables $i$ and $j$ such that $\sigma_{i,j}$ is unambiguous with respect to $\alpha$. In this regard, one particular choice can be ruled out immediately:
\begin{proposition}\label{prop:image_fixed_point}
Let $\alpha$ be a pattern, and let $i,j \in \var(\alpha)$, $i \neq j$. If $\sigma_{i,j}(\alpha)$ is a fixed point of a nontrivial morphism, then $\sigma_{i,j}$ is ambiguous with respect to $\alpha$.
\end{proposition}\noindent
For example, if we consider the pattern $\alpha_1 := 1 \cdot 2 \cdot 3 \cdot 4 \cdot 1 \cdot 4 \cdot 3 \cdot 2$ (which is not a fixed point) and define $\Sigma := \{a, b, c \}$, then $\sigma_{2,4}(\alpha_1)$ equals $abcbabcb$ (or any renaming thereof), which is a fixed point of the morphism $\phi$ given by $\phi(a) := abcb$ and $\phi(b) := \phi(c) := \varepsilon$. Thus, $\sigma_{2,4}$ is ambiguous with respect to $\alpha_1$. However, Proposition~\ref{prop:image_fixed_point} does not provide a characteristic condition on the ambiguity of $\sigma_{i,j}$, since $\sigma_{2,3}(\alpha_1) = abbcacbb$ is not a fixed point, but still $\sigma_{2,3}$ is ambiguous with respect to $\alpha_1$. Furthermore, while the ambiguity of $\sigma_{2,3}$ results from the fact that $\alpha_1$ contains the factors $2 \cdot 3$ and $3 \cdot 2$, and is therefore easy to comprehend, there are more difficult examples of morphisms $\sigma_{i,j}$ that are ambiguous although they do not lead to a morphic image that is a fixed point. This is illustrated by the example $\alpha_2 :=  1 \cdot 2 \cdot 3 \cdot 3 \cdot 4 \cdot 4 \cdot 1 \cdot 2 \cdot 3 \cdot 3 \cdot 4 \cdot 4 \cdot 2$. Here,  $\sigma_{2,4}(\alpha_1) = abccbbabccbbb$ again is not a fixed point, but $\sigma_{2,4}$ is nevertheless ambiguous with respect to $\alpha_2$, since the morphism $\tau$ given by $\tau(1) := abccb$, $\tau(2) := b$ and $\tau(3) := \tau(4) := \varepsilon$ satisfies $\tau(\alpha_2) = \sigma_{2,4}(\alpha_2)$. We therefore conclude that it seems not to be a straightforward task to find amendments that could turn Proposition~\ref{prop:image_fixed_point} into a characteristic condition.\par
We now show that Conjecture~\ref{conj:sigij} is correct for several types of patterns. To this end, we need the following simple sufficient condition on a pattern being a fixed point:
\begin{lemma}\label{lem:fixed_point}
Let $\alpha\in\mathbb{N}^{+}$. If there exists a variable $i\in\var(\alpha)$ such that
\begin{enumerate}
\item\label{lem-fixedpoint-condition1} $\varepsilon\not\in L_{i}$ and, for every $k\in L_{i}$, $R_{k}=\{i\}$, or
\item\label{lem-fixedpoint-condition2} $\varepsilon\not\in R_{i}$ and, for every $k\in R_{i}$, $L_{k}=\{i\}$,
\end{enumerate}
\noindent
then $\alpha$ is a fixed point of a nontrivial morphism.
\end{lemma}

Using this lemma, we can now establish a class of patterns for which Conjecture~\ref{conj:sigij} holds true. All variables in these patterns have the same number of occurrences and satisfy some additional conditions:
\begin{theorem}\label{thm:pair_not_covered_mn}
Let $m \in \mathbb{N}$, $m\geq2$. Let $\alpha \in \mathbb{N}^+$ be a pattern that is not a fixed point of a nontrivial morphism and satisfies, for every $x \in \var(\alpha)$, $|\alpha|_x = m$. If there are $i,j \in \var(\alpha)$, $i \neq j$, such that
\begin{itemize}
\item there is no $k\in\var(\alpha)$ with $\{i,j\}\subseteq L_{k}$ or $\{i,j\}\subseteq R_{k}$, and
\item $\alpha\neq \alpha_{1}\cdot \mathbf{i\cdot j}\cdot\alpha_{2}\cdot \mathbf{j\cdot i}\cdot\alpha_{3}$, $\alpha_{1},\alpha_{2},\alpha_{3}\in\mathbb{N}^*$,
\end{itemize}
then $\sigma_{i,j}$ is unambiguous with respect to $\alpha$.
\end{theorem}

\begin{proof}
Assume to the contrary that $\sigma_{i,j}$ is ambiguous. So, there exists a morphism $\tau: \mathbb{N}^+\to \Sigma^{*}$ satisfying $\tau(\alpha)=\sigma_{i,j}(\alpha)$ and, for some $x\in\var(\alpha)$, $\tau(x)\neq\sigma_{i,j}(x)$. Since $\sigma_{i,j}$ is a 1-uniform morphism, there exists a $k\in\var(\alpha)$ with $|\tau(k)|\geq 2$. Let $uv\sqsubseteq\tau(k)$, $u,v\in\Sigma$. Due to the fact that $k$ occurs $m$ times in $\alpha$, $\sigma_{i,j}(\alpha)=\tau(\alpha)=w_{1}\cdot uv\cdot w_{2}\cdot uv \cdot [\ldots]\cdot w_{m}\cdot uv\cdot w_{m+1}$ with, for every $q$, $1\leq q\leq m+1$, $w_{q}\in\Sigma^{*}$. We now consider the following cases:
\begin{itemize}
\item $\sigma_{i,j}(i)\neq u$ and $\sigma_{i,j}(i)\neq v$. This implies that there exist the variables $x_{1},x_{2}\in\var(\alpha)$, $x_{1},x_{2}\neq i$ and $x_{1},x_{2}\neq j$, such that $\alpha=\alpha_{1}\cdot x_{1}x_{2}\cdot\alpha_{2}\cdot x_{1}x_{2}\cdot[\ldots]\cdot\alpha_{m}\cdot x_{1}x_{2}\cdot\alpha_{m+1}$, for every $q$, $1\leq q\leq m+1$, $\alpha_{q}\in\mathbb{N}^*$, and $\sigma_{i,j}(x_{1})=u$ and $\sigma_{i,j}(x_{2})=v$. Due to $|\alpha|_{x_{1}}=|\alpha|_{x_{2}}=m$, the variables $x_1, x_2$ satisfy, for every $q$ with $1\leq q\leq m+1$, $x_{1},x_{2}\not\sqsubseteq\alpha_{q}$. This implies that $R_{x_{1}}=\{x_{2}\}$ and $L_{x_{2}}=\{x_{1}\}$. Then, according to Lemma~\ref{lem:fixed_point}, $\alpha$ is a fixed point of a nontrivial morphism, which is a contradiction to the assumption of the theorem.

\item $\sigma_{i,j}(i)= \sigma_{i,j}(j)= u$, and $u\neq v$. So, we assume that $\alpha=\alpha_{1}\cdot x_{1}x'\cdot\alpha_{2}\cdot x_{2}x'\cdot[\ldots]\cdot\alpha_{m}\cdot x_{m}x'\cdot\alpha_{m+1}$ with, $x'\in\var(\alpha)$ and, for every $q$, $1\leq q\leq m+1$, $x_{q}\in\var(\alpha)$, $\alpha_{q}\in\mathbb{N}^*$, and $\sigma_{i,j}(x_{q})=u$ and $\sigma_{i,j}(x')=v$. Additionally, since $\sigma_{i,j}(x')=v$ and $u\neq v$, we can conclude that $x'\neq i$ and $x'\neq j$. We now consider the following cases:

\begin{enumerate}
\item For every $q$, $1\leq q\leq m$, $x_{q}=i$. This implies, using the same reasoning as above, that $\alpha$ is a fixed point of a nontrivial morphism which is a contradiction.
\item There exists $q,q'$, $1\leq q,q'\leq m$ and $q\neq q'$, such that $x_{q}=i$ and $x_{q'}=j$. This means that $\{i,j\}\subseteq L_{x_{2}}$, which contradicts the first condition of the theorem.
\end{enumerate}

\item $\sigma_{i,j}(i)= v$, and $u\neq v$. The reasoning is analogous to that in the previous case.

\item $\sigma_{i,j}(i)=\sigma_{i,j}(j)=u$ and $v=u$. Hence, we may assume that $\alpha=\alpha_{1}\cdot x_{1}x'_{1}\cdot\alpha_{2}\cdot x_{2}x'_{2}\cdot[\ldots]\cdot\alpha_{m}\cdot x_{m}x'_{m}\cdot\alpha_{m+1}$ with, for every $q$, $1\leq q\leq m+1$, $\alpha_{q}\in\mathbb{N}^*$, $x_{q},x'_{q}\in\var(\alpha)$ and $\sigma_{i,j}(x_{q})=\sigma_{i,j}(x'_{q})=u$. Due to the conditions of the theorem, the factors $i\cdot i\cdot j$, $i\cdot j\cdot j$, $j\cdot i\cdot i$ and $j\cdot j\cdot i$ cannot be factors of $\alpha$. Moreover, it must be noticed that $u\cdot u\cdot u\not\sqsubseteq\tau(k)$; otherwise, since $\tau(\alpha)=\sigma_{i,j}(\alpha)$, then $|\alpha|_{i}>m$ or $\alpha_{j}>m$. This implies that $i\cdot j\cdot i$ and $j\cdot i\cdot j$ are not factors of $\alpha$. We now consider the following cases:

\begin{enumerate}
\item For every $q$, $1\leq q\leq m$, $x_{q}=i$ and $x'_{q}=j$. As a result, $R_{i}=\{j\}$ and $L_{j}=\{i\}$. According to Lemma~\ref{lem:fixed_point}, $\alpha$ is a fixed point of a nontrivial morphism.

\item For every $q$, $1\leq q\leq m$, $x_{q}=j$ and $x'_{q}=i$. Thus, $R_{j}=\{i\}$ and $L_{i}=\{j\}$, which, due to Lemma~\ref{lem:fixed_point}, again implies that $\alpha$ is a fixed point of a nontrivial morphism.

\item There exists a $q,q'$, $1\leq q,q'\leq m$ and $q\neq q'$, such that $x_{q}\cdot x'_{q}=i\cdot j$ and $x_{q'}\cdot x'_{q'}=j\cdot i$. This case contradicts the second condition of the theorem.

\item There exists a $q,q'$, $1\leq q,q'\leq m$ and $q\neq q'$, such that $x_{q}\cdot x'_{q}=i\cdot j$ and, $x_{q'}\cdot x'_{q'}=i\cdot i$ or $x_{q'}\cdot x'_{q'}=j\cdot j$. This means that $\{i,j\}\subseteq R_{i}$ or $\{i,j\}\subseteq L_{j}$, which is a contradiction to the first condition of the theorem.

\item There exists a $q,q'$, $1\leq q,q'\leq m$ and $q\neq q'$, such that $x_{q}\cdot x'_{q}=j\cdot i$ and, $x_{q'}\cdot x'_{q'}=i\cdot i$ or $x_{q'}\cdot x'_{q'}=j\cdot j$. This implies that $\{i,j\}\subseteq L_{i}$ or $\{i,j\}\subseteq R_{j}$, which contradicts the first condition of the theorem.

\item There exist $q,q'$, $1\leq q,q'\leq m$, $q'\neq q$, such that $x_{q}\cdot x'_{q}=i\cdot i$ and $x_{q'}\cdot x'_{q'}=j\cdot j$. Since $uu\sqsubseteq\tau(k)$ and due to the conditions of the theorem, it follows from $\tau(\alpha)=\sigma_{i,j}(\alpha)$ that $k\neq i$ and $k\neq j$. In other words, $\tau(i)\neq uu$ and $\tau(j)\neq uu$; otherwise, $|\tau(\alpha)|_{u}>|\sigma_{i,j}(\alpha)|_{u}$. Moreover, it must be noticed that if $\sigma_{i,j}(k)\sqsubseteq\tau(k)$, then this implies that there exists $x\in\var(\alpha)\setminus\{i,j\}$, with $\{i,j\}\subseteq L_{x}$ or $\{i,j\}\subseteq R_{x}$, which is a contradiction. Thus, $\sigma_{i,j}(k)\not\sqsubseteq\tau(k)$. Since $\tau(\alpha)=\sigma_{i,j}(\alpha)$, there must be a $k'\in\var(\alpha)$, $k'\neq k,i,j$, such that $\sigma_{i,j}(k)\sqsubseteq\tau(k')$, which means that $|\tau(k')|\geq2$, or we can extend the reasoning to other variables. Consequently, since $\tau(\alpha)=\sigma(\alpha)$, this discussion implies the existence of a $k''\in\var(\alpha)$, $k''\neq k,i,j$, such that $|\tau(k'')|\geq 2$, which, according to the above cases, leads to a contradiction.
\end{enumerate}
\end{itemize}
\end{proof}\noindent
We wish to point out that Theorem~\ref{thm:pair_not_covered_mn} does not only demonstrate the correctness of Conjecture~\ref{conj:sigij} for the given class of patterns, but additionally provides an efficient way of finding an unambiguous morphism $\sigma_{i,j}$. For example, we can immediately conclude from it that $\sigma_{1,4}$ is unambiguous with respect to our above example pattern $\alpha_1$. Furthermore, the theorem also holds for patterns with less than four different variables.\par
We now consider those patterns that are not a fixed point and, moreover, contain all of their variables exactly twice (note that some of these ``shortest'' patterns that are not fixed points are also studied in Theorem~\ref{thm:shortest_succinct_binary}). We wish to demonstrate that Theorem~\ref{thm:pair_not_covered_mn} implies the existence of an unambiguous $\sigma_{i,j}$ for \emph{every} such pattern. This insight is based on the following lemma:
\begin{lemma}\label{lem:greater_6}
Let $\alpha \in \mathbb{N}^+$ be a pattern with $|\var(\alpha)| > 6$ and, for every $x \in \var(\alpha)$, $|\alpha|_x = 2$. Then there exist $i,j\in\var(\alpha)$, $i\neq j$, such that
\begin{itemize}
\item there is no $k \in \var(\alpha)$ with $\{i,j\} \subseteq L_k$ or $\{i,j\} \subseteq R_k$, and
\item $\alpha \neq \alpha_1 \cdot \mathbf{i \cdot j} \cdot \alpha_2 \cdot \mathbf{j \cdot i} \cdot \alpha_3$, $\alpha_1, \alpha_2, \alpha_3 \in \mathbb{N}^*$.
\end{itemize}
\end{lemma}

\begin{proof}
Let $n := |\var(\alpha)|$. Since every variable occurs exactly twice in $\alpha$, it directly follows that, for every $x \in \var(\alpha)$, $|R_x| \leq 2$ and $|L_x| \leq 2$. By omitting the neighbourhood sets containing $\varepsilon$, we have at most $2n-2$ sets of size $2$. Besides, it can be verified with little effort that $\alpha$ contains at most $n-1$ different factors $i \cdot j$, $i,j \in \var(\alpha)$, $i\neq j$, such that $j \cdot i \sqsubseteq \alpha$ (e.\,g., for $n:=4$, $\alpha := 1 \cdot 2 \cdot 3 \cdot 4 \cdot 4 \cdot 3 \cdot 2 \cdot 1$ has $3$ different factors $i\cdot j$, $i,j \in \var(\alpha)$, $i\neq j$, satisfying $j \cdot i \sqsubseteq \alpha$). Assume to the contrary that, for every $i,j\in\var(\alpha)$, one of the following cases is satisfied:
\begin{itemize}
\item there exists a $k \in \var(\alpha)$ with $\{i,j\} \subseteq L_k$ or $\{i,j\} \subseteq R_k$, or
\item $\alpha = \alpha_1 \cdot \mathbf{i \cdot j} \cdot \alpha_2 \cdot \mathbf{j \cdot i} \cdot \alpha_3$, $\alpha_1, \alpha_2, \alpha_3 \in \mathbb{N}^*$.
\end{itemize}
As mentioned above, the maximum number of pairs that are covered by the first case is $2n-2$, and for the second case it is $n-1$. On the other hand, since $|\var(\alpha)|=n$, there exist ${n \choose 2}$ different pairs of variables. However, for $n>6$, we have ${n\choose 2} > (2n-2) + (n-1)$, which contradicts the assumption.
\end{proof}
Hence, whenever a pattern $\alpha$ is not a fixed point, the conditions of Theorem~\ref{thm:pair_not_covered_mn} are automatically satisfied if $\alpha$ contains at least seven distinct variables and all of its variables occur exactly twice. Using a less elegant reasoning than the one on Lemma~\ref{lem:greater_6}, we can extend this insight to all such patterns over at least four distinct variables. This yields the following result:
\begin{theorem}\label{thm:shortest_succinct}
Let $\alpha\in \mathbb{N}^+$ be a pattern with $|\var(\alpha)| > 3$ and, for every $x \in \var(\alpha)$, $|\alpha|_x = 2$. If $\alpha$ is not a fixed point of a nontrivial morphism, then there exist $i,j \in \var(\alpha)$, $i \neq j$, such that $\sigma_{i,j}$ is unambiguous with respect to $\alpha$.
\end{theorem}

Theorem~\ref{thm:shortest_succinct} does not only directly prove the correctness of Conjecture~\ref{conj:sigij} for all patterns that contain all their variables exactly twice, but it also allows a large set of patterns to be constructed for which the Conjecture holds true as well. This construction is specified as follows:
\begin{theorem}\label{thm:shortest_as_factor}
Let $\alpha:=\alpha_{1}\cdot\beta\cdot\alpha_{2}$ and $\gamma:=\alpha_{1}\cdot\alpha_{2}$ be patterns with $\alpha_{1},\alpha_{2},\beta \in \mathbb{N}^*$, such that
\begin{itemize}
\item $\gamma$ and $\beta$ are not a fixed point of a nontrivial morphism,
\item $|\var(\gamma)|>3$ and, for every $x \in \var(\gamma)$, $|\gamma|_x = 2$, or $|\var(\beta)|>3$ and, for every $x \in \var(\beta)$, $|\beta|_x = 2$, and
\item $\var(\gamma)\cap\var(\beta)=\emptyset$.
\end{itemize}
Then there exist $i,j \in \var(\alpha)$, $i \neq j$, such that $\sigma_{i,j}$ is unambiguous with respect to $\alpha$.
\end{theorem}

In the remainder of this section, we shall not directly address the morphism $\sigma_{i,j}$ any longer. Hence, we focus on Conjecture~\ref{conj:variable}, and we use an approach that differs quite significantly from those above: We consider words that cannot be morphic images of a pattern under any ambiguous 1-uniform morphism, and we construct suitable morphic preimages from these words. This method yields another major set of patterns for which Conjecture~\ref{conj:variable} is satisfied.\par
Our corresponding technique is based on the well-known concept of \emph{de Bruijn sequences}. Since de Bruijn sequences are cyclic, which does not fit with our subject, we introduce a non-cyclic variant:
\begin{definition}\label{def:n_DB_non_cyclic}
A \emph{non-cyclic De Bruijn sequence} (of order $n$) is a word over a given alphabet $\Sigma$ (of size $k$) for which all possible words of length $n$ in $\Sigma^*$ appear exactly once as factors of this sequence. We denote the set of all non-cyclic De Bruijn sequences of order $n$ by $B'(k,n)$.
\end{definition}\noindent
For example, the word $w_0 := aabacbbcca$ is a non-cyclic de Bruijn sequence in $B'(3,2)$ if we assume $\Sigma := \{ a, b, c \}$.\par
It can now be easily understood that a non-cyclic de Bruijn sequence cannot be a morphic image of any pattern under ambiguous 1-uniform morphisms:
\begin{theorem}\label{thm:un_factors_once}
Let $\Sigma$ be an alphabet, and let $\alpha \in \mathbb{N}^+$ be a pattern satisfying, for every $x \in \var(\alpha)$, $|\alpha|_x \geq 2$. Let $\sigma: \mathbb{N}^* \to \Sigma^*$ be a 1-uniform morphism such that, for every $u_1 u_2 \sqsubseteq \sigma(\alpha)$, $u_{1}, u_{2} \in \Sigma$, the factor $u_1 u_2$ occurs in $\sigma(\alpha)$ exactly once. Then $\sigma$ is unambiguous with respect to $\alpha$.
\end{theorem}

This insight implies that every pattern that can be mapped by a 1-uniform morphism to a de Bruijn sequence necessarily is not a fixed point, and thus, fits with Conjecture~\ref{conj:variable}:
\begin{corollary}\label{cor:fp_factors_once}
Let $\Sigma$ be an alphabet, and let $\alpha \in \mathbb{N}^+$ be a pattern satisfying, for every $x \in \var(\alpha)$, $|\alpha|_x\geq2$. Let $\sigma: \mathbb{N}^*\to \Sigma^*$ be a 1-uniform morphism such that, for every $u_{1} u_{2} \sqsubseteq\sigma(\alpha)$, $u_1, u_2 \in \Sigma$, the factor $u_1 u_2$ occurs in $\sigma(\alpha)$ exactly once. Then $\alpha$ is not a fixed point of a nontrivial morphism.
\end{corollary}

We now show how we can construct patterns that fit with the requirements of Theorem~\ref{thm:un_factors_once} and Corollary~\ref{cor:fp_factors_once}:
\begin{definition}\label{def:Pi_DB}
Let $\Sigma := \{ a_1, a_2, \ldots, a_k \}$. Let $B'(k,2)$ be the set of non-cyclic de Bruijn sequences of order $2$ over $\Sigma$. Then $\Pi_{DB}(k) \subseteq \mathbb{N}^*$ is the set of all patterns that can be constructed as follows: For every $w \in B'(k,2)$ and every letter $a_j$ in $w$, all $n_j$ occurrences of $a_j$ are replaced by $\lfloor n_j / 2 \rfloor$ different variables from a set $N_j :=\{x_{j_1}, x_{j_2}, \ldots, x_{j_{\lfloor n_j / 2 \rfloor}}\} \subseteq \mathbb{N}$, such that the following conditions are satisfied:
\begin{itemize}
\item for every $x \in N_j$, $|\alpha|_x > 1$,
\item for all $i, i'$, $1 \leq i,i' \leq k$, with $i \neq i'$, $N_i \cap N_{i'} = \emptyset$, and
\item for all $i$, $1 \leq i \leq k$, the variables in $N_i$ are assigned to occurrences of $a_i$ in a way such that the resulting pattern is in canonical form.
\end{itemize}
\end{definition}\noindent
For instance, with regard to our above example word $w_0 = aabacbbcca \in B'(3,2)$, Definition~\ref{def:Pi_DB} says that, e.\,g., the pattern $1 \cdot 1 \cdot 2 \cdot 3 \cdot 4 \cdot 2 \cdot 2 \cdot 4 \cdot 4 \cdot 3$ is contained in $\Pi_{DB}(3)$.\par
From this construction, it follows that Conjecture~\ref{conj:variable} holds true for every pattern in $\Pi_{DB}(k)$:
\begin{theorem}\label{thm:Pi_DB}
Let $\Sigma := \{ a_1, a_2, \ldots, a_k \}$, $k \geq 3$. Then, for every $\alpha \in \Pi_{DB}(k)$,
\begin{itemize}
\item $\var(\alpha)$ contains at least $k+1$ elements, and
\item there exists a 1-uniform morphism $\sigma: \mathbb{N}^* \to \Sigma^*$ that is unambiguous with respect to $\alpha$.
\end{itemize}
\end{theorem}

\begin{proof}
We begin this proof with the first statement of the theorem: It is obvious that there are $k^2$ different words of length $2$ over $\Sigma$. The shortest word that contains $k^2$ factors of length $2$ has length $k^2 + 1$, which means that this is the length of any word $w \in B'(k,2)$. Thus, there must be at least one letter in $w$ that has at least $\lceil (k^2 + 1)/k \rceil$ occurrences. Since we assume $k \geq 3$, this means that this letter has at least $4$ occurrences. From Definition~\ref{def:Pi_DB} it then follows that this letter is replaced by at least two different variables when a pattern $\alpha \in \Pi_{DB}(k)$ is generated from $w$. Since all other letters in $w$ must be replaced by at least one variable, this shows that $|\var(\alpha)| \geq k + 1$.\par
Concerning the second statement, we define $\sigma$ by, for every $j$, $1 \leq j \leq k$, and for every $x \in N_j$, $\sigma(x) := a_j$. Thus, $\sigma$ is 1-uniform, and $\sigma(\alpha) \in B'(k,2)$. This implies that, for every $u_1 u_2 \sqsubseteq \sigma(\alpha)$, $u_1, u_2 \in \Sigma$, the factor $u_1 u_2$ occurs in $\sigma(\alpha)$ exactly once. Consequently, according to Theorem~\ref{thm:un_factors_once}, $\sigma$ is unambiguous with respect to $\alpha$.
\end{proof}

We conclude this paper with a statement on the cardinality of $\Pi_{DB}(k)$, demonstrating that the use of de Bruijn sequences indeed leads to a rich class of patterns $\alpha$ with unambiguous 1-uniform morphisms, and that these morphisms, in general, can even have a target alphabet of size much less than $\var(\alpha) - 1$ (as featured by Theorem~\ref{thm:Pi_DB}):
\begin{theorem}\label{thm:card_Pi_DB}
Let $k \in \mathbb{N}$. Then $|\Pi_{DB}(k)| \geq k!^{(k-1)}$, and, for every $\alpha \in \Pi_{DB}(k)$,
\begin{displaymath}
|\var(\alpha)| = (k - 1) \lfloor k/2 \rfloor + \lfloor (k+1)/2 \rfloor \, .
\end{displaymath}
\end{theorem}

\paragraph{Acknowledgements}
The authors wish to thank the anonymous referees for their helpful remarks and suggestions.

\bibliographystyle{eptcs}

\end{document}